\documentclass[final,1p,times]{elsarticle}
\usepackage{mathrsfs}
\usepackage[english]{babel}
\usepackage{amsfonts,amsmath,amsxtra,amsthm,amssymb,latexsym}
\usepackage{verbatim}

\biboptions{square,comma,numbers}

\newtheorem{theorem}{Theorem}[section]

\newtheorem{lemma}[theorem]{Lemma}

\newtheorem{remark}[theorem]{Remark}
\newtheorem{example}[theorem]{Example}

\numberwithin{equation}{section}

\def\P{{\mathsf P}}

\def\H{\mathscr H}
\def\K{\mathscr K}

\def\D{\mathscr D}
\def\V{\mathsf V}
\def\R{\mathscr R}

\def\G{\mathscr G}

\def\W{\mathsf W}

\def\Dsf{\mathsf D}
\def\Gammasf{\mathsf\Gamma}
\def\ve{\mathsf v}
\def\y{\mathsf y}
\def\x{\mathsf x}
\def\Y{\mathsf Y}

\def\uno{\mathsf 1}

\def\E{\mathsf {E}}

\def\RE{\mathbb R}
\def\CO{{\mathbb C}}

\def\ccdot{\!\cdot\!}
\def\min{\text{\rm min}}
\def\max{\text{\rm max}}
\def\Omegasf{\mathsf\Omega}

\journal{Journal of Functional Analysis}
\begin{document}
\begin{frontmatter}
\author{Andrea Posilicano}
\ead{posilicano@uninsubria.it}
\address{DiSAT - Sezione di Matematica, Universit\`a dell'Insubria, Como, Italy}
\title{On the many Dirichlet Laplacians on a non-convex polygon 
and their approximations by point interactions}

\begin{abstract}
By Birman and Skvortsov it is known that if $\Omegasf$ is a 
planar
curvilinear polygon with $n$ non-convex 
corners then the Laplace operator 
with domain $H^2(\Omegasf)\cap H^1_0(\Omegasf)$ is a closed symmetric
operator with deficiency indices
$(n,n)$. 
Here we provide a Kre\u\i n-type resolvent formula for any self-adjoint extensions of 
such an  operator, i.e. for the set of self-adjoint non-Friedrichs Dirichlet Laplacians on $\Omegasf$, and show that any element in this set is the norm resolvent limit of a suitable sequence of Friedrichs-Dirichlet Laplacians with $n$ point interactions.
\end{abstract}

\begin{keyword} 
Dirichlet Laplacians\sep Point Interactions 
\sep Self-Adjoint Extensions \sep Kre\u \i n's Resolvent Formula
\MSC 47B25 \sep 35J05 \sep 35J10 \sep 81Q10

\end{keyword}
\end{frontmatter}

\begin{section}{Introduction.}
Since their rigorous mathematical definition by Berezin and Faddeev
\cite{[BF]} as self-adjoint extensions of the Laplacian restricted to 
smooth functions with compact support disjoint from a finite set 
 in 
$\RE^d$, $d\le 3$, point perturbations of the Laplacian have
attracted a lot of attention and have been used in a wide 
range of applications, as 
the huge list of references provided in \cite{[AGHH]} shows. Successively point perturbations of the Dirichlet Laplacian on a bounded 
domain have been defined in a similar way, see \cite{[CS]}, \cite{[BFM]}, \cite{[EM]}. In this case, since functions in the
domain of the Dirichlet Laplacian vanish at the boundary, 
points perturbations can not be placed there. Nevertheless one 
could try to
put point-like perturbations at the boundary by moving the 
 points supporting the  perturbation towards the 
boundary while
increasing the interactions strengths,
so to compensate the vanishing of the functions. 
However it is not clear how to implement this procedure, 
because there is no universal behavior for the 
functions in the operator domain 
in a neighborhood of the boundary. 
For example if $\Omegasf\subset\RE^2$ is 
a planar bounded domain which either 
has a regular  (i.e without corners) boundary  or is convex, then the self-adjoint
Friedrichs-Dirichlet Laplacian on $L^2(\Omegasf)$ has domain 
$H^2(\Omegasf)\cap H^1_0(\Omegasf)$. Here $H^k(\Omegasf)$ denotes the usual Hilbert-Sobolev space of $k$-th order and the subscript means ``zero at the boundary''. Thus, by the (dense) 
inclusions
$C^\infty_0(\Omegasf)\subset H^2(\Omegasf)\cap H^1_0(\Omegasf)
\subset C_0(\Omegasf)$, 
there is no minimal vanishing rate for $u(\x)$ as $\x$ 
approaches the boundary. The situation changes if
one considers a planar non-convex polygon. Indeed in this case 
the Friedrichs-Dirichlet
Laplacian has a domain that is strictly larger than $H^2(\Omegasf)\cap
H^1_0(\Omegasf)$ and for any function $u$ in such a domain one has 
$u(\x)\sim \xi_u\, s_\ve(\x)\|\x-\ve\|^{\pi/\omega}$ when $\|\x-\ve\|\ll 1$, where
$\ve$ is the vertex at a non-convex
corner, $\omega>\pi$ is the measure of the interior angle at $\ve$,  and $0<s_\ve(\x)\le 1$. This 
indicates that it should be possible
to renormalize the value of $u$
at $\ve$ by considering the limit  of 
$\|\x-\ve\|^{-\frac{\pi}{\omega}}s_\ve(\x)^{-1}u(\x)$ as $\x\to\ve$. 
Indeed such a procedure works and in the case of an arbitrary
point perturbation of the Friedrichs-Dirichlet Laplacian on a planar
polygon $\Omegasf$ with $n$ non-convex corners,  
the limit operator, 
as the $n$ points supporting the perturbations converge to
the $n$ non-convex vertices, turns out to be a well defined 
self-adjoint
operator: it coincides with a self-adjoint extension 
of the closed symmetric operator 
(which by \cite{[BS]} 
has deficiency indices $(n,n)$ ) given by 
the Laplace operator on $H^2(\Omegasf)\cap H^1_0(\Omegasf)$. \par
The proof we give in this paper follows the reverse
path. \par At first in Section 2 we provide a Kre\u\i n's resolvent formula for any self-adjoint extensions of  the Laplace operator on 
$H^2(\Omegasf)\cap
H^1_0(\Omegasf)$, $\Omegasf$ a bounded non-convex curvilinear polygon (unknown to the author, some similar results had been 
given in last section of the 
unpublished paper \cite{[DM]}; we thank Mark Malamud for 
the communication). Here we work in a operatorial setting; however we profit by some known results obtained by a more PDE-oriented approach, the literature on the subject being abundant: see e.g. \cite{[Ko]}, \cite{[MMPS]}, \cite{[G2]}, \cite{[Da]}, \cite{[G]}, \cite{[Ni]}, \cite{[NP]}, \cite{[KMR]}, \cite{[BK]} and references therein.  
The operator domain of any of the self-adjoint extensions is contained in the kernel of the unique continuous extension to the domain of the maximal Laplacian of the trace (evaluation) operator along the boundary, and the functions in the operator domains still satisfy Dirichlet's boundary condition $u(\x)=0$, provided
$\x$ is not the vertex of a non-convex corner. Thus such
family of self-adjoint extensions defines a set of non-Friedrichs
Dirichlet Laplacians, the Friedrichs-Dirichlet Laplacian being the
only one satisfying Dirichlet's boundary conditions also at the
vertices of the non-convex corners. 
\par 
In Section 3 we define an arbitrary $n$-point perturbation of
the Friedrichs-Dirichlet Laplacian $\Delta^F_\Omegasf$ by considering all self-adjoint extensions of the symmetric operator given by the restriction of $\Delta^F_\Omegasf$ to the set of function vanishing at $n$ points contained in $\Omega$. Then we provide a corresponding Kre\u\i n's resolvent formula (see \cite{[BFM]} and \cite{[EM]} for similar results) and we show that, if the points supporting  the perturbations converge to the
non-convex vertices of $\Omegasf$, while the coupling strength is renormalized
according to the vanishing rate of the functions in  $\D(\Delta^F_\Omegasf)$, 
these self-adjoint operators converge in norm resolvent sense
to the self-adjoint
extensions provided in Section 2 (see Theorem \ref{convergence}).\par In the Appendix, we
collect, following the approach developed in \cite{[Po1]}-\cite{[Po4]}, some results about self-adjoint extensions of
symmetric operators that we need in the proofs. In particular we give a simple convergence criterion for sequences of extensions (see Lemma \ref{taun}). 

\begin{subsection}{Notations.}{\ }
\par\noindent
$\bullet$ $\D(L)$, $\K(L)$, $\R(L)$, $\rho(L)$ denote
the domain, kernel,
range and resolvent set of a  closed linear operator $L$  
on an Hilbert space $\H$; 
\par\noindent
$\bullet$ $\|\phi\|_L=(\|L\phi\|^2_\H+\|\phi\|^2_\H)^{1/2}$ 
denotes the graph norm on $\D(L)$;
\par\noindent
$\bullet$
$L|{\mathcal V}$ denotes the restriction of $L$ to 
${\mathcal V}\subset\D(L)$;\par\noindent
$\bullet$ $L^2(\Omegasf)$ denotes the Hilbert space of
square-integrable functions on the open domain 
$\Omegasf$ with scalar product $\langle
u,v\rangle_{L^2(\Omegasf)}=
\int_\Omegasf 
\bar u(\x)v(\x)\,d\x$;
\par\noindent
$\bullet$ The dot $\cdot$ denotes the scalar product on $\CO^n$, 
i.e. $\xi\ccdot\zeta=\sum_{k=1}^{n}\bar\xi_k\zeta_k$;
\par\noindent
$\bullet$ $\CO^n_\Pi\equiv\R(\Pi)$ denotes the subspace corresponding
to the orthogonal projector $\Pi:\CO^n\to\CO^n$. By a slight abuse of notation
we use the same symbol $\Pi$ also to denote the 
injection $\Pi|\CO_\Pi^n:\CO_\Pi^n\to\CO^n $ and the surjection 
$(\Pi|\CO_\Pi^n)^*: \CO^n\to\CO^n_\Pi$;
\par\noindent
$\bullet$  $\E(\CO^n)$ denotes the bundle $p:\E(\CO^n)\to \P(\CO^n)$, where 
$\P(\CO^n)$ is the set of orthogonal projectors on $\CO^n$ and $p^{-1}(\Pi)$ is the set of symmetric operators on $\CO^n_\Pi$;
\par\noindent
$\bullet$  $c$ denotes a
generic strictly positive constant which can change from line to line.
\end{subsection}
\end{section}
\begin{section}{Dirichlet Laplacians on a non-convex planar curvilinear polygon.}
Let $\mathsf\Omegasf\subset\RE^2$ be
a bounded open Lipschitz domain.
This means that in the neighborhood of any of its point $\Omegasf$ is
below the graph of a Lipschitz function and such a graph coincides with its boundary $\Gammasf$. \par
We denote by $\Delta_\Omegasf$ the distributional
Laplace operator on $\Omegasf$ and we define
$$
\Delta_\Omegasf^\max:\D(\Delta_\Omegasf^\max)\subset
L^2(\Omegasf)\to L^2(\Omegasf)\,,\qquad \Delta_\Omegasf^\max u:=
\Delta_\Omegasf u\,,
$$
where
$$
\D(\Delta_\Omegasf^\max):=
\{u\in L^2(\Omegasf)\,:\, \Delta_\Omegasf u\in L^2(\Omegasf)\}\,.
$$
We denote by $C^{\infty}(\bar\Omegasf)$ the set of functions on
$\bar\Omegasf$, the closure of $\Omegasf$, which are restriction to
$\bar\Omegasf$ of smooth functions with compact support
on $\RE^2$ and we denote by $H^k(\Omegasf)$
the Sobolev-Hilbert
space given by closure of
$C^{\infty}(\bar\Omegasf)$ with respect to the norm defined by
$$
\|u\|_{H^k(\Omegasf)}^2=\sum_{0\le \alpha_1+\alpha_2\le k}
\left\| \partial_1^{\alpha_1}\partial_2^{\alpha_2} 
u\right\|^2_{L^2(\Omegasf)}
\,.
$$
By Sobolev embedding theorem one has, for any $\alpha\in (0,1)$,
\begin{equation}\label{emb}
H^2(\Omegasf)\subseteq C^{\alpha}(\bar\Omegasf)
\end{equation}
and
\begin{equation}\label{emb2}
\forall u\in H^2(\Omegasf)\,,\ \forall\x,\y\in\bar\Omegasf\,,\quad |u(\x)-u(\y)|\le c\,\|u\|_{H^2(\Omegasf)}\|\x-\y\|^\alpha\,.
\end{equation}
Analogously $H^k_0(\Omegasf)$ denotes
the closure of $C_c^{\infty}(\Omegasf)$, the set of smooth function with
compact support on $\Omegasf$, with respect to the
same norm.
The space $H^1_0(\Omegasf)$ can be equivalently defined by
$$
H^1_0(\Omegasf):=\{u\in H^1(\Omegasf)\,:\, \gamma_0 u=0\}\,,
$$
where  
$$
\gamma_0:H^1(\Omegasf)\to
L^{2}(\Gammasf)
$$
is the unique continuous 
linear map such that
$$
\forall u\in C^\infty(\bar\Omegasf)\,,\quad\forall\,
\x\in\Gammasf \,,\qquad
 \gamma_0 u\,(\x)=u\,(\x)\,.
$$
There is a standard, well known way to define a self-adjoint Dirichlet
Laplacian on  $L^2(\Omegasf)$
:
since the symmetric sesquilinear form
$$
F_\Omegasf:H^1_0(\Omegasf)\oplus H^1_0(\Omegasf)\subset L^2(\Omegasf)\oplus
L^2(\Omegasf)\to\CO\,,\quad
F_\Omegasf(u,v):=
\langle\nabla u,\nabla v\rangle_{L^2(\Omegasf)}
$$
is closed and positive, by Friedrichs' extension theorem
there exists an unique self-adjoint operator
$$
\Delta_\Omegasf^F:\D(\Delta_\Omegasf^\max)\cap H^1_0(\Omegasf)
\subset L^2(\Omegasf)\to L^2(\Omegasf)\,,\qquad
\Delta_\Omegasf^Fu=\Delta_\Omegasf u\,,
$$
such that
$$
\forall u\in \D(\Delta_\Omegasf^\max)\cap H^1_0(\Omegasf)\,,
\quad\forall v\in H^1_0(\Omegasf)\,,\qquad F_\Omegasf(u,v)
=-\langle \Delta^F_\Omegasf u,v\rangle_{L^2(\Omegasf)}\,.
$$
Moreover
$$\D(\Delta^F_\Omegasf)\equiv\D(\Delta_\Omegasf^\max)\cap
H^1_0(\Omegasf)$$
is dense in $H^1_0(\Omegasf)$,
$0\in \rho
(\Delta^F_\Omegasf)$, $-\Delta^F_\Omegasf$ has a compact
resolvent, and its spectrum consists of
an infinite sequence
$$\lambda_1(\Omegasf)<\lambda_2(\Omegasf)\le\lambda_3(\Omegasf)\le
\dots$$  of strictly positive
eigenvalues each having finite multiplicity,
$\lambda_1(\Omegasf)$ being simple. We call $\Delta^F_\Omegasf$ the
{\it Friedrichs-Dirichlet Laplacian}.
\par
In the case $\Omegasf$ is piecewise regular there is another way to produce a 
self-adjoint Dirichlet
Laplacian on  $L^2(\Omegasf)$. Thus from now on we suppose that
$\Omegasf\subset \RE^2$ is a planar bounded open curvilinear polygon
(cups points are not allowed) 
which coincides with a planar polygon in the neighborhood of any
(eventual) non-convex corner. 
\par
Let us  recall the following Caccioppoli-type regularity estimate:
\begin{equation}\label{cacc}
\forall u\in H^2(\Omegasf)\cap H^1_0(\Omegasf)\,,\qquad \|u\|_{H^2(\Omegasf)}\le c\, \|\Delta_\Omegasf u\|_{L^2(\Omegasf)}\,.
\end{equation}
The proof of such an estimate,
for general elliptic second order
differential operator
on a class of bounded
open sets which includes curvilinear polygons,
can be found in \cite{[LU]} (see Lemma 8.1, Chapter 3, Section 8); for the Laplace operator on polygons a much simpler
proof is given in \cite{[G]}, Theorem 2.2.3.\par
By \eqref{cacc}, since $H^2(\Omegasf)\cap H^1_0(\Omegasf)$ is closed in $H^2(\Omegasf)$,
(see e.g. \cite{[G]}, Theorem 1.6.2), the linear operator
$$
\Delta_\Omegasf^\circ:H^2(\Omegasf)\cap H^1_0(\Omegasf)\subset
L^2(\Omegasf)\to L^2(\Omegasf)\,,\qquad \Delta_\Omegasf^\circ u:=
\Delta_\Omegasf u
$$
is closed. Moreover, by Green's formula for curvilinear polygons (see \cite {[G2]}, Lemma 1.5.3.3), $\Delta^\circ_\Omegasf$ is symmetric. Thus a natural question arises: is $\Delta^\circ_\Omegasf$
self-adjoint? Equivalently:
does $\Delta^\circ_\Omegasf$ coincide with $\Delta^F_\Omegasf$ ?
If $\Omegasf$ had a regular
boundary with no corners then
$\D(\Delta_\Omegasf^\max)\cap H_0^1(\Omegasf)=H^2(\Omegasf)\cap
H^1_0(\Omegasf)$,
i.e. $\Delta^\circ_\Omegasf=\Delta^F_\Omegasf$. Otherwise
the answer depends on the shape of $\Omegasf$.
Indeed if $\Omegasf$ is a curvilinear polygon
then, as it has been proven in \cite {[BS]}, the deficiency indices of $\Delta^\circ_\Omegasf$
are both equal to $n$,
the number of non-convex corners of $\Omegasf$.
In this case $\D(\Delta_\Omegasf^\max)\cap H_0^1(\Omegasf)\not
=H^2(\Omegasf)\cap
H^1_0(\Omegasf)$
is an immediate consequence of the fact that the function
$$
u(r,\theta)=r^{\beta}\,\sin \beta\theta\,,\quad \beta:
=\frac{\pi}{\omega}\,,$$
belongs to $H^1(\W)$, where $\W$ is the
wedge $$\W=\{\x\equiv(r\cos\theta,r\sin\theta)\,:\, 0< r<1\,,\ 0<\theta<\omega\}\,,$$
is in $\D(\Delta_\W^\max)$ since $\Delta_\W u=0$, 
but fails to be in
$H^2(\W)$ when $\omega>\pi$.
\par
From now on we will suppose that $n>0$ so that
$$\D(\Delta^\circ_\Omegasf)\equiv 
H^2(\Omegasf)\cap H^1_0(\Omegasf)
\subsetneq \D(\Delta^{\max}_\Omegasf)\cap H^1_0(\Omegasf)\equiv 
\D(\Delta^F_\Omegasf)\,.$$ 
Since $\Delta^\max_\Omegasf$ is the
adjoint of the restriction of $\Delta_\Omegasf$ to
$C^\infty_c(\Omegasf)$, one has
$({\Delta^\circ_\Omegasf})^*\subset
\Delta^\max_\Omegasf$ and so any  self-adjoint
extension of $\Delta^\circ_\Omegasf$ acts on the
functions in its domain as the distributional
Laplacian.\par 
Let us at first characterize
$\K({(\Delta^\circ_\Omegasf)}^*)$.  To this end
we need the extension $\hat \gamma_0$ of 
$\gamma_0$ to $\D(\Delta_\Omegasf^\max)$ provided in \cite{[G2]}, Theorem
1.5.3.4, and  \cite{[G]},
Theorem 1.5.2:  there
exits an unique continuous
map
$$
\hat\gamma_0:
\D(\Delta_\Omegasf^\max)\to\oplus_{i=1}^m \tilde H^{-\frac12}
(\Gammasf_i)\,,
$$
which coincides with $\gamma_0$ on
$\D(\Delta_\Omegasf^\max)\cap
H^1(\Omegasf)$.
Here $\tilde
H^{-\frac12}(\Gammasf_i)$ denote  Hilbert spaces of distributions on the smooth curves $\Gammasf_i$, $i=1,\dots,m$, which union, together with their endpoints (i.e. the vertices of $\Omegasf$), give $\Gammasf$. We do not need here the precise definition of $\tilde
H^{-\frac12}(\Gammasf_i)$, see the quoted references for the details. For our purposes it suffices to say that 
$$
\hat\gamma_0 u=0\quad\iff\quad\forall i=1,\dots,m\,,\ 
\forall\varphi_i\in C_c^\infty(\Gammasf_i)\,,\quad \langle(\hat\gamma_0 u)_i,\varphi_i\rangle=0\,.
$$
Then, by Lemma 2.3.1 and
Theorem 2.3.3 in \cite{[G]}, one has the following
\begin{theorem}\label{kernel}
$$\K(({\Delta^\circ_\Omegasf})^*)
=\K(\Delta_\Omegasf^\max)\cap\K(\hat \gamma_0)\,.$$
\end{theorem}
As we already said before, contrarily to the case of a domain $\Omegasf$ either convex or
with a regular boundary, the kernel of ${(\Delta^\circ_\Omegasf)}^*$
is not trivial. Indeed in \cite{[BS]} it is shown that 
\begin{equation}\label{number}
\text{\rm dim}\,\K(({\Delta^\circ_\Omegasf})^*)=
\text{\rm number of non-convex corners of $\Omegasf$}.
\end{equation}
In order to better characterize $\K(({\Delta^\circ_\Omegasf})^*)$ we
introduce some more definitions. 
Let $\V=\{\ve_1,\dots,\ve_n\}$ be the set of vertices at the non-convex
corners of $\Omegasf$ and, for any $\ve_k\in\V$, let 
$\omega_k>\pi$ 
denote the measure of the corresponding interior 
angle. We define the wedge
\begin{align*}
\W^R_k:=\Omegasf\cap \Dsf^R_k
\equiv
\{\x_k\equiv(r_k\cos \theta_k, r_k\sin\theta_k)\in\RE^2\,:\,
0< r_k<R\,,\ 0<\theta_k<\omega_k\}\,,
\end{align*}
where $\Dsf^R_k$ denotes the disk of radius $R$ centered at $\ve_k$; we choose $R$ small enough to have 
$\W^R_i\cap \W^R_k=\emptyset$, $i\not=k$.
On any disk $\Dsf_k$ centered at $\ve_k$ we consider the functions
$u_k^\pm\in \K(\Delta_{\Dsf_k}^\max)$ defined by
$$
u_k^\pm(r_k,\theta_k)=
\frac{1}{\sqrt\pi}\ r_k^{\pm\beta_k}\sin \beta_k\theta_k\,,\quad
\beta_k:=\frac{\pi}
{\omega_k}\,,
$$
and we take $f\in C^{1,1}(\RE_+)$, i.e. $f$ is differentiable with a Lipschitz derivative, such that $0\le f\le 1$,
$f(r)=1$ if $0<r\le R/3$ and $f(r)=0$
if $r\ge 2R/3$.
With such a choice we have
$fu_k^\pm\in L^2(\Omegasf)$ and, since supp$(fu_k^\pm)
\subset\W^R_k$,
the functions $fu_k^\pm$ are $L^2(\Omegasf)$-orthogonal and 
thus linearly independent.
\begin{lemma}\label{sigma} Let us define
$$
s_k:=f u^+_k\,,\quad \sigma_k:=f u^-_k\,, \quad
g_k:=
\sigma_k+(-{\Delta_\Omegasf^F})^{-1}\Delta_\Omegasf
\sigma_k\,.
$$
Then\par\noindent
1) $$s_k\in\D(\Delta_\Omegasf^F)\,,\qquad
\sigma_k\in\D((\Delta_\Omegasf^\circ)^*)\cap\K(\hat\gamma_0)\,;$$
2) $g_k$ is the unique
function in $\K(({\Delta_\Omegasf^\circ})^*)$ such that
$$g_k-\sigma_k
\in\D(\Delta_\Omegasf^F);$$
3) the  $g_k$'s are linearly independent;\par\noindent
4)$$
\langle
g_i,(-\Delta_\Omegasf^F) s_k\rangle_{L^2(\Omegasf)}=\delta_{ik}\,.
$$
\end{lemma}
\begin{proof} By $u_k^+\in C^\infty(\W^R_k)
\cap H^1(\W^R_k)$ one has $s_k\in \D(\Delta_\Omegasf^{\max})\cap H^1_0(\Omegasf)$ .
\par
By $\sigma_k\in\K(\hat\gamma_0)$ there follows 
$g_k\in\D(\Delta_\Omegasf^{\max})\cap \K(\hat\gamma_0)$. Hence $g_k\in\K((\Delta_\Omegasf^\circ)^*)$ by Theorem \ref{kernel}. This also shows that $\sigma_k\in\D((\Delta_\Omegasf^\circ)^*)$. Proof of point 2 is then completed by $\K(\Delta_\Omegasf^F)=\{0\}$. 
\par
Take $c_1,\dots,c_n$ such that
$\sum_{k=1}^{n}c_k g_k=0$. Then
$$
({\Delta_\Omegasf^F})^{-1}\Delta_\Omegasf\sum_{k=1}^{n}c_k
\sigma_k=\sum_{k=1}^{n}c_k \sigma_k\,.
$$
This gives $c_1=\dots=c_n=0$, since the
$\sigma_k\,$'s are linearly independent and do
not belong to $\D(\Delta_\Omegasf^F)$.
Thus point 3 is proven.
\par
As regards point 4, let us pose
$\W_k:=\W^{2R/3}_k\backslash \W^{R/3}_k$. Then
\begin{align*}
&\langle
g_k,\Delta_\Omegasf^F s_k\rangle_{L^2(\Omegasf)}=
\langle \sigma_k,\Delta_\Omegasf^F s_k\rangle_{L^2(\W_k)}
-\langle\Delta_\Omegasf \sigma_k,s_k\rangle_{L^2(\W_k)}
\\
=&\int_{\W_k}f u^-_k
\left(f''u_k^++\left(1+\frac{2\pi}{\omega_k}\right)
\frac{1}{r}\,f'u^+_k\right)\,dx
\\
&-\int_{W_k}\left(f''u_k^-+\left(1-\frac{2\pi}{\omega_k}\right)
\frac{1}{r}\,f'u^-_k\right)f u^+_k\,dx\\
=&\frac{2}{\omega_k}\int_{R/3}^{2R/3}2f'f\,dr
\int_0^{{\omega_k}}\sin^2\frac{\pi}{\omega_k}\,\theta\ d\theta
=-\frac{2}{\pi}\,\int_0^{\pi}\sin^2\theta\ d\theta=-1\,.
\end{align*}
\end{proof}
For notational convenience let us define the $\CO^n$-valued functions
$$
s\equiv(s_1,\cdots,s_n)\,,\quad
\sigma\equiv(\sigma_1,\cdots,\sigma_n)\,,
\quad g\equiv(g_1,\cdots,g_n)\,.
$$
By Lemma 2.3.6 and (the proof of) Theorem 2.3.7 in
\cite{[G]}, (\ref{number}) can be specified:
\begin{theorem}\label{ker}
For any $u\in \K(({\Delta^\circ_\Omegasf})^*)$ there exist an 
unique
$\xi_u\in\CO^n$ such that
$u=g\ccdot\xi_u$.
\end{theorem}
In order to use the results given in the 
Appendix we need a more precise characterization of
$\D(\Delta_\Omegasf^{\max})\cap H_0^1(\Omegasf)$ i.e. of 
$\D(\Delta_\Omegasf^F)$:
\begin{theorem}\label{domF}
$$
\D(\Delta_\Omegasf^F)
=\{u\in L^2(\Omegasf)\,:\, u=u_\circ+s\ccdot\zeta_u\,,\ u_\circ\in \D(\Delta_\Omegasf^\circ)\,,
\, \zeta_{u}\in\CO^n\}\,.
$$
\end{theorem}
\begin{proof} By Theorem \ref{ker} and point 4 in
Lemma \ref{sigma}, the linearly independent functions 
$\Delta_\Omegasf^Fs_k\,$ are not
  orthogonal to $\K(({\Delta_\Omegasf^\circ})^*)$. Thus, given 
$u\in \D(\Delta_\Omegasf^F)$,  
the decomposition
$L^2(\Omegasf)=
\R (\Delta_\Omegasf^\circ)\oplus \K(({\Delta_\Omegasf^\circ})^*)$ 
implies
that there exist
unique $\tilde u_\circ\in \D(\Delta_\Omegasf^\circ)$ and
$\zeta_{u}\in\CO^n$ such that
$$\Delta_\Omegasf^F u=\Delta_\Omegasf^\circ\tilde  u_\circ+
\Delta_\Omegasf^F s\ccdot\zeta_u\,.$$
\end{proof}
Next we introduce a convenient map $\tau^\V_\Omegasf$ such that 
$\D({\Delta^\circ_\Omegasf})=\K(\tau^\V_\Omegasf)$:
\begin{lemma}\label{tau} Let 
$$
\tau^\V_\Omegasf:\D(\Delta_\Omegasf^F)\to\CO^n\,,
\quad \left(\tau^\V_\Omegasf u\right)_k
:=\frac{\sqrt{\pi^3}}{4}\,(2+\beta_k)\,
\lim_{R\downarrow  0}\ \frac{1}{R^{\beta_k}}\,\langle u\rangle_{\W^R_k}
\,,
$$
where $\langle u\rangle_{\W^R_k}$ denotes the mean of $u$ over 
the wedge $\W^R_k$. Then $\tau^\V_\Omegasf$ is well defined, 
continuous, surjective 
and $\K(\tau^\V_\Omegasf)=\D({\Delta^\circ_\Omegasf})$.
\end{lemma}
\begin{proof} By Theorem \ref{domF} $\D(\Delta_\Omegasf^F)=
\D(\Delta_\Omegasf^\circ)+{\mathcal V}$, where both 
$\D(\Delta_\Omegasf^\circ)$ and 
${\mathcal V}$ are closed subspaces of 
$\D(\Delta_\Omegasf^F)$,  and $\D(\Delta_\Omegasf^\circ)\cap
{\mathcal V}=\{0\}$. Therefore the map   
$$
P_\circ:\D(\Delta_\Omegasf^F)\to\CO^n\,,\quad P_\circ u=\zeta_u\,,
$$ 
given by the composition of the continuous projection onto
$\mathcal V$ with the identification map 
giving ${\mathcal V}\simeq\CO^n$, 
is  continuous. To conclude we show
that $\tau^\V_\Omegasf=P_\circ$, i.e. $\tau^\V_\Omegasf u=\zeta_u$.
By Theorem \ref{domF} one has $u=u_\circ+s\ccdot\zeta_u$. 
Thus, by using (\ref{emb2}) with $\alpha\in (\beta_k,1)$, 
one has
$$
|\tau^\V_\Omegasf u_\circ|\le\,c\,
\lim_{R\downarrow  0}\ \frac{1}{R^{\beta_k}}
\left(\frac{\omega_k}{2}\,R^2\right)^{-1}
\int_{W^R_k} |u_\circ(x)|\,dx\le c\,\lim_{R\downarrow  0}\, 
R^{\alpha-\beta_k}=0\,, 
$$
while
\begin{align*}
&(\tau^\V_\Omegasf s\ccdot\zeta_u)_k
=\frac{\sqrt{\pi^3}}{4}\,(2+\beta_k)\,(\zeta_u)_k\,
\lim_{R\downarrow  0}\ \frac{1}{R^{\beta_k}}\,
\langle s_k\rangle_{\W^R_k}\\
=&(\zeta_u)_k\,\frac{\beta_k}{2}\int_0^{\omega_k}
\sin\beta_k\theta\ d\theta\ \lim_{R\downarrow  0}\,
 \frac{2+\beta_k}{R^{2+\beta_k}}\,
\int_0^{R}r^{1+\beta_k}dr
=(\zeta_u)_k\,,
\end{align*}
and the proof is done.
\end{proof}
Combining Theorem \ref{domF} and Lemma \ref{tau} with the results  provided in Appendix 
in the case $A=\Delta_\Omegasf^F$ and 
$\tau=\tau_\Omegasf^\V$,  we can obtain easily a resolvent formulae for 
all self-adjoint extensions of $\Delta_\Omegasf^\circ$. To this end we give the following
\begin{lemma} \label{G0} Let 
$$
G_z:\CO^n\to L^2(\Omegasf)\,,\qquad G_z^*:L^2(\Omegasf)\to\CO^n\,,\qquad z\in\rho(\Delta_\Omegasf^F)\,,
$$
be defined by  
$$
G_z:=\left(\tau^\V_\Omegasf(-\Delta_\Omegasf^F+\bar z)^{-1}\right)^*\,.
$$
Then
$$
G_z\xi=
\sigma\ccdot\xi-
(-\Delta_\Omegasf^F+z)^{-1}(-\Delta_\Omegasf+z)\sigma\ccdot\xi
$$
and
$$
G_z^*u=
\langle\sigma,u\rangle_{L^2(\Omegasf)}-
\langle(-\Delta_\Omegasf^F+z)^{-1}(-\Delta_\Omegasf+z)\sigma,
u\rangle_{L^2(\Omegasf)}\,.
$$
\end{lemma}
\begin{proof} By Lemma \ref{sigma}, Theorem \ref{domF} and 
Lemma \ref{tau}  
one has
\begin{align*}
&\langle g,-\Delta_\Omegasf^Fu\rangle_{L^2(\Omega)}=
\langle g,-\Delta_\Omegasf^\circ u_\circ -
\Delta_\Omegasf^F s\ccdot\zeta_u
\rangle_{L^2(\Omega)}\\
=&\langle{(-\Delta_\Omegasf^\circ )}^*g,
u_\circ\rangle_{L^2(\Omega)}+\langle g,-
\Delta_\Omegasf^Fs
\rangle_{L^2(\Omega)}\ccdot\zeta_u\\
=&\zeta_u=\tau^\V_\Omegasf u\,.
\end{align*}
Thus
\begin{align*}
&\langle G_0\xi,u\rangle_{L^2(\Omegasf)}=
\xi\ccdot\tau_\Omegasf{(-\Delta_\Omegasf^F)}^{-1}u
=\xi\ccdot\langle g,u\rangle_{L^2(\Omega)}
\,,
\end{align*}
i.e.
$$
G_0\,\xi=g\ccdot\xi=
\sigma\ccdot\xi+
(-\Delta_\Omegasf^F)^{-1}\Delta_\Omegasf\sigma\ccdot\xi\,.
$$
By (\ref{GR1}) one has then
\begin{align*}
G_z\xi=&(\uno-z(-\Delta_\Omegasf^F+z)^{-1})G_0\xi
=(g-z(-\Delta_\Omegasf^F+z)^{-1}g)\ccdot\xi\\
=&(\sigma-(-\Delta_\Omegasf^F+z)^{-1}(-\Delta_\Omegasf+z)\sigma)\ccdot\xi\,.
\end{align*}
\end{proof}
Notice that 
in next theorem we use the extension, 
denoted by the same symbol, of $(\tau^\V_\Omegasf)_k$ to 
functions coinciding away from $\ve_k$ with functions in 
$\D(\Delta_\Omegasf^F)$ .
\begin{theorem}\label{estensioniOmega2}
Any self-adjoint extension of $\Delta_\Omegasf^\circ$ is of the kind
$$
\Delta_{\Omegasf}^{\Pi,\Theta}:\D(\Delta_{\Omegasf}^{\Pi,\Theta})
\subset L^2(\Omegasf)\to L^2(\Omegasf)\,, \quad
\Delta_{\Omegasf}^{\Pi,\Theta}u:=\Delta_\Omegasf u\,,$$
\begin{align*}
\D(\Delta_{\Omegasf}^{\Pi,\Theta}):=
\{u\in L^2(\Omegasf)\,:\,u=u_0+g\ccdot\xi_u\,,\ u_0\in 
\D(\Delta_{\Omegasf}^F)\,,\ \xi_u\in\CO^n_\Pi\,,\ 
\Pi\hat \tau^\V_\Omegasf u=\Theta\xi_u\}\,,
\end{align*}
where $(\Pi,\Theta)\in\E(\CO^n)$ and
$$
(\hat \tau^\V_\Omegasf u)_k=(\tau^\V_\Omegasf (u-(\xi_u)_k\,g_k))_k\,.
$$
Moreover
\begin{align*}
(-\Delta_\Omegasf^{\Pi,\Theta}+z)^{-1}
=(-\Delta^F_{\Omegasf}+z)^{-1}+G_z
\Pi\,(\Theta+\Pi\, \Gamma_z\Pi)^{-1}
\Pi\, G_{\bar z}^*\,,
\end{align*}
where
\begin{align*}
&(\Gamma_z)_{ij}=\left(z\|\sigma_i\|^2_{L^2(\Omegasf)}+
\|(-\Delta_\Omegasf^F)^{-\frac12}\Delta_{\Omegasf}\sigma_i\|^2_{L^2(\Omegasf)}\right)\,
\delta_{ij}\\
-&\langle(-\Delta_\Omegasf^F+z)^{-1}(-\Delta_\Omegasf+z)\sigma_i,
(-\Delta_\Omegasf+z)
\sigma_j\rangle_{L^2(\Omegasf)}
\end{align*}
\end{theorem}
\begin{proof}
By Theorem \ref{estensioni}  
any $u=u_0+G_0\xi_u=u_0+g\ccdot\xi_u$ in the domain of a 
self-adjoint extension has to satisfy the boundary condition
$\Pi\tau^\V_\Omegasf u_0=\tilde\Theta\xi_u$, 
for some $(\Pi,\tilde\Theta)\in\E(\CO^n)$. 
Thus
\begin{align*}
(\tau_\Omegasf^\V u_0)_i=(\tau_\Omegasf^\V  
(u-g_i(\xi_u)_i))_i-\sum_{j\not=i}\left(\tau_\Omegasf^\V  
(g_j(\xi_u)_j)\right)_i
=(\hat\tau_\Omegasf^\V  u)_i-(\Lambda\xi_u)_i\,,
\end{align*}
where, by Lemma \ref{G0},
\begin{equation}\label{lambda}
\Lambda_{ij}=\langle(-\Delta_\Omegasf^F)^{-1}
\Delta_\Omegasf\sigma_i,
\Delta_\Omegasf\sigma_j\rangle_{L^2(\Omegasf)}
(1-\delta_{ij})\,.
\end{equation}
Moreover, by (\ref{gamma1}) and Lemma \ref{G0}, one has
\begin{align*}
&z(G^*_0G_z)_{ij}=(\tau^\V_{\Omegasf}(G_0-G_z))_{ij}\\
=&(\tau^\V_{\Omegasf}(
(-\Delta_\Omegasf^F)^{-1}\Delta_\Omegasf\sigma
+(-\Delta_\Omegasf^F+z)^{-1}(-\Delta_\Omegasf+z)\sigma))_{ij}\\
=&(G_0^*\Delta_\Omegasf\sigma)_{ij}+
(G_z^*(-\Delta_\Omegasf+z)\sigma)_{ij}\\
=&\Lambda_{ij}+\left(\langle\sigma_i,\Delta_\Omegasf\sigma_i
\rangle_{L^2(\Omegasf)}+\langle(-\Delta^F_\Omegasf)^{-1}
\Delta_\Omegasf\sigma_i,\Delta_\Omegasf\sigma_i
\rangle_{L^2(\Omegasf)}\right)\,\delta_{ij}\\
&+
\langle\sigma_i,(-\Delta_\Omegasf+z)\sigma_i\rangle_{L^2(\Omegasf)}
\,\delta_{ij}\\
&-\langle(-\Delta_\Omegasf^F+z)^{-1}(-\Delta_\Omegasf+z)\sigma_i,
(-\Delta_\Omegasf+z)
\sigma_j\rangle_{L^2(\Omegasf)}\\
=&\Lambda_{ij}+(\Gamma_z)_{ij}\,.
\end{align*}
The proof is then concluded by taking $\tilde\Theta=\Theta
-\Pi \Lambda\Pi$.
\end{proof}

Since $g_k-\sigma_k\in\D(\Delta^F_\Omegasf)$, 
Theorem \ref{estensioniOmega2} admits an alternative version:
\begin{theorem}\label{estensioniOmega}
Any self-adjoint extension of $\Delta_\Omegasf^\circ$ is of the kind
$$
\tilde\Delta_{\Omegasf}^{\Pi,\Theta}:
\D(\tilde\Delta_{\Omegasf}^{\Pi,\Theta})
\subset L^2(\Omegasf)\to L^2(\Omegasf)\,, \quad
\tilde\Delta_{\Omegasf}^{\Pi,\Theta}u:=\Delta_\Omegasf u\,,$$
\begin{align*}
\D(\tilde\Delta_{\Omegasf}^{\Pi,\Theta}):=
\{u\in L^2(\Omegasf)\,:\,u=u_0+\sigma\ccdot\xi_u\,,\ u_0\in 
\D(\Delta_{\Omegasf}^0)\,,\ \xi_u\in\CO^n_\Pi\,,\ 
\Pi\tau^\V_\Omegasf u_0=\Theta\xi_u\}\,,
\end{align*}
where $(\Pi,\Theta)\in\E(\CO^n)$.
Moreover
\begin{align*}
(-\tilde\Delta_\Omegasf^{\Pi,\Theta}+z)^{-1}
=(-\Delta^F_{\Omegasf}+z)^{-1}+G_z
\Pi\,(\Theta+\Pi\, \tilde\Gamma_z\Pi)^{-1}
\Pi\, G_{\bar z}^*\,,
\end{align*}
where
\begin{align*}
&(\tilde\Gamma_z)_{ij}=\left(z\|\sigma_i\|^2_{L^2(\Omegasf)}-
\langle\sigma_i,\Delta_\Omegasf\sigma_i\rangle_{L^2(\Omegasf)}\right)
\,\delta_{ij}\\
-&\langle(-\Delta_\Omegasf^F+z)^{-1}(-\Delta_\Omegasf+z)\sigma_i,
(-\Delta_\Omegasf+z)
\sigma_j\rangle_{L^2(\Omegasf)}
\end{align*}
\end{theorem}
\begin{proof}
By Theorem \ref{estensioni} and Lemma \ref{G0} 
any $u$ in the domain of a self-adjoint extension 
of $\D(\Delta_\Omegasf^\circ)$ is of the kind
$u=\tilde u_0+g\ccdot\xi_u$, $\tilde 
u_0\in\D(\Delta_\Omegasf^F)$, 
$\xi_u\in\CO^n_\Pi$, $\Pi\tau^\V_\Omegasf 
\tilde u_0=\tilde\Theta\xi_u$, 
for some $(\Pi,\tilde\Theta)\in\E(\CO^n)$. By the
definition of $g$ (see Lemma \ref{sigma}), one has 
$u=u_0+\sigma\ccdot\xi_u$, $u_0=\tilde 
u_0+(-{\Delta_\Omegasf^F})^{-1}\Delta_\Omegasf
\sigma\ccdot\xi_u$ and, by Lemma \ref{G0}, 
\begin{align*}
\tau_\Omegasf^\V \tilde u_0=&\tau_\Omegasf^\V  u_0
-\tau_\Omegasf^\V
(-\Delta_\Omegasf^F)^{-1}\Delta_\Omegasf\sigma\ccdot\xi_u
=\tau_\Omegasf^\V u_0-G^*_0\Delta_\Omegasf\sigma\ccdot\xi_u\\
=&\tau_\Omegasf^\V u_0-\tilde\Lambda\xi_u
\,,
\end{align*}
where
\begin{align*}
\tilde\Lambda_{ij}=&
\langle\sigma_i,\Delta_\Omegasf\sigma_i\rangle_{L^2(\Omegasf)}
\,\delta_{ij}+\langle(-\Delta_\Omegasf^F)^{-1}
\Delta_\Omegasf\sigma_i,
\Delta_\Omegasf\sigma_j\rangle_{L^2(\Omegasf)}\\
=&\tilde\Lambda_{ii}\delta_{ij}+\Lambda_{ij}\,.
\end{align*}
By noticing that $\tilde\Gamma_{ij}=
\Gamma_{ij}-\tilde\Lambda_{ii}\delta_{ij}$, 
the proof is then concluded by taking $\tilde\Theta=\Theta
-\Pi \tilde\Lambda\Pi$.
\end{proof}

\begin{remark} {\rm Since both $\sigma_k$ and $g_k$, $1\le k\le n$, 
belong to $\K(\hat\gamma_0)$ (see Lemma \ref{sigma}), 
all the self-adjoint extensions 
of $\Delta_\Omegasf^\circ$ have domains contained in 
$\K(\hat\gamma_0)$, the only one with domain contained 
in $\K(\gamma_0)$ being Friedrichs' 
Laplacian $\Delta_\Omegasf^F$. Thus we can interpret the set 
of all self-adjoint extensions 
of $\Delta_\Omegasf^\circ$ different from $\Delta_\Omegasf^F$ as the set of self-adjoint, 
non-Friedrichs' Dirichlet Laplacians on $L^2(\Omegasf)$.  }
\end{remark}

\begin{remark}{\rm 
Notice that if both $\Pi$ and $\Theta$ are diagonal, 
then both the boundary conditions 
$\Pi\hat \tau_\Omegasf^\V u=\Theta\xi_u$ and $\Pi\tau_\Omegasf^\V u_0=\Theta\xi_u$ appearing in the previous
theorems 
are local, i.e. they do not couple values of $u$ at 
different vertices.  }
\end{remark}

\begin{example}  {\rm The prototypical example is provided by the curvilinear polygon 
$\Omegasf=\W$, where $\W$ denotes the non-convex wedge
$$
\W=\{\x\equiv(r\cos\theta,r\sin\theta)\,:\, 0< r<R\,,\ 0<\theta<\omega\}\,,\quad\omega\in(\pi,2\pi)\,.
$$
In this case $\K({(\Delta_\W^\circ)}^*)$ is one 
dimensional and by Theorem 2.1 $g$ is the unique
(up to the multiplication by a constant) solution of the boundary value problem

\begin{align*}
\begin{cases}
\Delta_\W^{\max}g(r,\theta)=0\,,&\\
g(r,0)=g(r,\omega)=g(R,\theta)=0\,,\quad r\not=0\,.&
\end{cases}
\end{align*}
Thus
$$
g(r,\theta)=\frac{1}{\sqrt{\pi}}\,
\left(\frac{1}{r^{\beta}}-
\frac{r^{\beta}}{R^{2\beta}}\right)\sin\beta\theta\,,\quad\beta=\frac{\pi}{\omega}\,.
$$
Similarly $G_z:\CO\to L^2(\W)$ 
acts as the multiplication by the
function $g_z$ which solves the boundary value problem
\begin{align*}
\begin{cases}
{\Delta_\W^{\max}}g_z(r,\theta)=z\,g_z(r,\theta)\,,&\\
g(r,0)=g(r,\omega)=g(R,\theta)=0\,,\quad r\not=0\,.&
\end{cases}
\end{align*}
Thus
\begin{align*}
&g_z(r,\theta)\\
=&\frac{1}{\sqrt{\pi}}\,\left(\frac{\sqrt z}{2}\right)
^{\beta}\,\Gamma\left(1-\beta\right)\,
\left(J_{-\beta}(\sqrt z \,r)-\frac{J_{-\beta}(\sqrt z \,R)}
{J_{\beta}(\sqrt z \,R)}\,
J_{\beta}(\sqrt z \,r)\right)\sin\beta\theta\,,
\end{align*}
where Re$(\sqrt z)>0$, $\Gamma(x)$ denotes Euler's gamma function at
$x$ and $J_{\nu}$ denotes the Bessel function of order 
$\nu$. Here the constants are chosen in order to have 
$g_z\to g$ as $z\to 0$.
Then
\begin{align*}
&\Gamma_z=z\langle g,g_z\rangle_{L^2(\W)}=
\frac{1}{R^{2\beta}}+\left(\frac{z}{4}\right)^{\beta}
\frac{\Gamma(-\beta)}{\Gamma(\beta)}\,
\frac{J_{-\beta}(\sqrt z \,R)}{J_{\beta}(\sqrt z \,R)}\,.
\end{align*}
By Theorem \ref{estensioniOmega2} the set of 
self-adjoint
extensions of $\Delta^\circ_\W$ different from $\Delta^F_\W$ is
parametrized by $\theta\in\RE$. Any of such extensions 
has resolvent $R_z^\theta$ with kernel
$$
R_z^\theta(\x,\y)=R^F_z(\x,\y)+\left(\theta+
\left(\frac{z}{4}\,\right)^{\beta}
\frac{\Gamma(-\beta)}{\Gamma(\beta)}\,
\frac{J_{-\beta}(\sqrt z \,R)}{J_{\beta}(\sqrt z \,R)}\right)^{-1}
g_z(\x)g_{z}(\y)\,,
$$
where $R^F_z$ denotes the resolvent of the 
Friedrichs-Dirichlet Laplacian. }

\end{example}

\end{section}
\begin{section}{Approximation by Friedrichs-Dirichlet Laplacians with point interactions}
\begin{subsection}{Laplacians with point interactions}
Let $\Delta_{\Omegasf}^F$ be the Friedrichs-Dirichlet Laplacian on
$\Omegasf$ as defined in the previous section and, given the discrete
set $\Y=\{\y_1,\dots, \y_n\}\subset\Omegasf$, we define the 
linear map 
$$
\tau^\Y_\Omegasf:\D(\Delta_{\Omegasf}^F)\to\CO^n\,,\qquad 
(\tau^\Y_\Omega u)_k:=u(\y_k)\,,\quad k=1,\dots,n\,.
$$
By $u(\y_k)=u_\circ(\y_k)+s(\y_k)\ccdot\zeta_u$ and 
(\ref{cacc}) such a linear map is continuous with respect with the
graph norm of $\Delta_{\Omegasf}^F$, 
is evidently surjective and has a
dense (in $L^2(\Omegasf)$) kernel. Thus we can apply the 
results provided in
the Appendix to write down all the self-adjoint extensions of
the symmetric operator $\Delta_{\Y,\Omegasf}^\circ$ 
given by restricting $\Delta_{\Omegasf}^F$ to
the functions that vanish at $\Y$. \par
Let us denote by $g_\Omegasf(z;\cdot,\cdot)$ Green's function of $-\Delta_{\Omegasf}^F+z$, so that
$$
g_\Omegasf(z;\x,\y)=g(z;\x,\y)-h_\Omegasf(z;\x,\y)\,,
$$
where
$$
g(0;\x,\y)=\frac{1}{2\pi}\,\ln\frac{1}{\|\x-\y\|}\,,
$$
\begin{align*}
g(z;\x,\y)=\frac{1}{2\pi}\,K_0(\sqrt z\,\|\x-\y\|)\,,\quad 
\text{\rm Re}(\sqrt z)>0\,,
\end{align*}
$K_0$ the Macdonald (or modified Hankel) function, 
and $h_\Omegasf(z;\cdot,\y)$ 
solves the inhomogeneous Dirichlet boundary value problem
\begin{equation*}
\begin{cases}(-\Delta_\Omegasf^F+z)h_\Omegasf(z;\x,\y)=0\,,
&\quad \x\in \Omegasf\\
h_\Omegasf(z;\x,\y)=g(z;\x,\y)\,,&\quad \x\in \Gammasf\,.
\end{cases}
\end{equation*}
\begin{remark} {\rm One has (see e.g. \cite{[GR]}, formula 8.447.3) 
$$
K_0(\sqrt z\,\|\x-\y\|)=\ln\frac{1}{\|\x-\y\|}
-\ln\frac{\sqrt z}{2}+\psi(1)+o(\sqrt z\,\|\x-\y\|)\,,
$$
where $\psi$ is Euler's psi function.\par 
Since $\Omegasf$ satisfies the exterior cone condition and $g(z;\cdot,\cdot)$ is continuous outside the diagonal, by regularity of solutions of boundary value problems for elliptic equations with continuous boundary data (see e.g. \cite{[K]}, Corollary 7.4.4), one has $h_\Omegasf(z;\cdot,\y)\in C^{\infty}(\Omegasf)\cap C(\bar\Omegasf)$ for any $\y\in\Omegasf$.\par 
By Theorem 21 in \cite{[Sw]} one has
\begin{equation}\label{Sweers}
\frac{1}{c}\,g_\Omegasf(0;\x,\y)\le\ln\left(1+w(\x,\y)\,
\frac{u_1(\x)\,u_1(\y)}{\|\x-\y\|^2}\right)\le
 c\,g_\Omegasf(0;\x,\y)\,,
\end{equation}
where 
\begin{equation*}
w(\x,\y):=\begin{cases}
\min\left(\left(\frac{d(\x)}{u_1(\x)}\right)^2,
\left(\frac{d(\y)}{u_1(\y)}\right)^2
\right)\,,& \hat d(\x)<R\,,\ \hat d(\y)<R\,,\\
\max\left(\left(\frac{d(\x)}{u_1(\x)}\right)^2,
\left(\frac{d(\y)}{u_1(\y)}\right)^2
\right)\,,& \check d(\x)<R\,,\ \check d(\y)<R\,,\\
1\,,&{\text{\rm otherwise}}\,,
\end{cases} 
\end{equation*}
$u_1$ is the first eigenfunction of $-\Delta^F_\Omegasf$, 
$d$ is the 
distance form the boundary, $\hat d$ the distance
from the set of the vertices at the convex corners, $\check d$ is the distance
from the set of the vertices at the non-convex corners, and $2R$ is  smaller than the distance between any couple of vertices.}
\end{remark}
Defining
$$
G^\Y_z:\CO^n\to L^2(\Omegasf)\,,\qquad 
(G^\Y_z)^*:L^2(\Omegasf)\to\CO^n \,,\qquad z\in\rho(\Delta_\Omegasf^F)
$$
by
$$
G^\Y_z:=\left(\tau^\Y_\Omegasf(-\Delta_\Omegasf^F+\bar z)^{-1}\right)^*
\,,
$$
one has
$$
(G^\Y_z\xi)(\x)=\sum_{i=1}^{n}g_\Omegasf(z;\x,\y_i)\,\xi_i
$$
and
$$
((G^\Y_z)^*u)_k=\langle g_\Omegasf(z;\cdot,\y_k),u
\rangle_{L^2(\Omegasf)}\,.
$$
By the results provided in the Appendix one obtains 
the following 
\begin{theorem}\label{estensioniY2}
Any self-adjoint extension of $\Delta^\circ_{\Y,\Omegasf}$ 
is of the kind
$$
\Delta_{\Y,\Omegasf}^{\Pi,\Theta}:\D(\Delta_{\Y,\Omegasf}^{\Pi,\Theta})
\subset L^2(\Omegasf)\to L^2(\Omegasf)\,, \quad
\Delta_{\Y,\Omegasf}^{\Pi,\Theta}u:=\Delta_\Omegasf ^Fu_0\,,
$$
\begin{align*}
\D(\Delta_{\Y,\Omegasf}^{\Pi,\Theta}):=
\{u\in L^2(\Omegasf)\,:\,
u=u_0+G_0^\Y\xi_u\,,\ u_0\in\D(\Delta^F_{\Omegasf})\,,\ 
\xi_u\in\CO^n_\Pi\,,\ \Pi\hat\tau_\Omegasf^\Y
u=\Theta\xi_u\},
\end{align*}
where $(\Pi,\Theta)\in\E(\CO^n)$ and
$$
\left(\hat\tau_\Omegasf^\Y u\right)_k:=
\lim_{\x\to\y_k}\,\left(u(\x)-\frac{(\xi_u)_k}{2\pi}\,
g_\Omegasf(0;\x,\y_k)\right)\,.
$$
Moreover
\begin{align*}
(-\Delta_{\Y,\Omegasf}^{\Pi,\Theta}+z)^{-1}
=(-\Delta^F_{\Omegasf}+z)^{-1}+G^\Y_z
\Pi\,(\Theta+\Pi\,\Gamma^\Y_z\Pi)^{-1}
\Pi\, (G^\Y_{\bar z})^*\,,
\end{align*}
where
\begin{align*}
\left(\Gamma_z^\Y\right)_{ij}:=&\left(
\frac{1}{2\pi}\,\left(\ln\left(\frac{\sqrt
  z}{2}\right)-\psi(1)\right)-h_\Omegasf(0;\y_i,\y_i)
+h_\Omegasf(z;\y_i,\y_i) \right)
\,\delta_{ij}\\
&-g_\Omegasf(z;\y_i,\y_j)\,(1-\delta_{ij})\,.
\end{align*}
\end{theorem}
\begin{proof}
By Theorem \ref{estensioni}  
any $u=u_0+G_0^\Y\ccdot\xi_u$ 
in the domain of a self-adjoint extension 
of $\D(\Delta_{\Y,\Omegasf}^\circ)$ has to satisfy the 
boundary conditions $\Pi\tau^\V_\Omegasf 
u_0=\tilde\Theta\xi_u$, 
for some $(\Pi,\tilde\Theta)\in\E(\CO^n)$. Since
\begin{align*}
(\tau_\Omegasf^\Y u_0)_i
=\lim_{\x\to\y_i}\,\left(u(\x)-g(0;\x,\y_i)\,(\xi_u)_i\right)
-\sum_{j\not=i}g_\Omegasf (0;\y_i,\y_j)\,(\xi_u)_j\,,
\end{align*} 
one has
$$
\tau_\Omegasf^\Y u_0=\hat\tau_\Omegasf^\Y u-\Lambda^\Y\xi_u\,,
$$
where 
\begin{equation}\label{lambdaY}
\Lambda^\Y_{ij}:=g_\Omegasf(0;\y_i,\y_j)\,(1-\delta_{ij})
\,.
\end{equation} 
Moreover
\begin{align*}
&\left(z(G^\Y_0)^*G^\Y_z\right)_{ij}=\left(\tau^\Y_\Omegasf(G^\Y_0-G^\Y_z)\right)_{ij}\\
=&\left(\lim_{\x\to\y_i}\,
\left(g(0;\x,\y_i)-g(z;\x,\y_i)\right)-h_\Omegasf(0;\y_i,\y_i)+
h_\Omegasf(z;\y_i,\y_i)\right)\,\delta_{ij}\\
&+
\left(g_\Omegasf(0;\y_i,\y_j)-g_\Omegasf(z;\y_i,\y_j)\right)\,(1-\delta_{ij})\\
=&\Lambda^\Y_{ij}+\left(\Gamma^\Y_z\right)_{ij}
\,,
\end{align*}
Thus the proof is concluded by posing 
$\tilde\Theta=\Theta-\Pi\Lambda^\Y\Pi$.
\end{proof}
By the decomposition $g_\Omegasf=g+h_\Omegasf$ the previous theorem admits (the proof being of the same kind) an alternative version which provides results analogous to the ones given in \cite{[BFM]} and \cite{[EM]}. 
\begin{theorem}\label{estensioniY}
Any self-adjoint extension of $\Delta^\circ_{\Y,\Omegasf}$ 
is of the kind
$$
\check\Delta_{\Y,\Omegasf}^{\Pi,\Theta}:\D(\check\Delta_{\Y,\Omegasf}^{\Pi,\Theta})
\subset L^2(\Omegasf)\to L^2(\Omegasf)\,, \quad
\check\Delta_{\Y,\Omegasf}^{\Pi,\Theta}u:=\Delta_\Omegasf ^Fu_0\,,
$$
\begin{align*}
\D(\check\Delta_{\Y,\Omegasf}^{\Pi,\Theta}):=
\{u\in L^2(\Omegasf)\,:\,
u=u_0+G_0^\Y\xi_u\,,\ u_0\in\D(\Delta^F_{\Omegasf})\,,\ 
\xi_u\in\CO^n_\Pi\,,\ \Pi\check\tau_\Omegasf^\Y
u=\Theta\xi_u\},
\end{align*}
where $(\Pi,\Theta)\in\E(\CO^n)$ and
$$
\left(\check\tau_\Omegasf^\Y u\right)_k:=
\lim_{\x\to\y_k}\,\left(u(\x)-\frac{(\xi_u)_k}{2\pi}\,
\ln\frac{1}{\|\x-\y_k\|}\right)\,.
$$
Moreover
\begin{align*}
(-\check\Delta_{\Y,\Omegasf}^{\Pi,\Theta}+z)^{-1}
=(-\Delta^F_{\Omegasf}+z)^{-1}+G^\Y_z
\Pi\,(\Theta+\Pi\,\check\Gamma^\Y_z\Pi)^{-1}
\Pi\, (G^\Y_{\bar z})^*\,,
\end{align*}
where
\begin{align*}
\left(\check\Gamma_z^\Y\right)_{ij}:=&\left(
\frac{1}{2\pi}\,\left(\ln\left(\frac{\sqrt
  z}{2}\right)-\psi(1)\right)
+h_\Omegasf(z;\y_i,\y_i) \right)
\,\delta_{ij}\\
&-g_\Omegasf(z;\y_i,\y_j)\,(1-\delta_{ij})\,.
\end{align*}
\end{theorem}
\begin{remark}{\rm
Notice that if both $\Pi$ and $\Theta$ are diagonal, 
then  both the boundary conditions 
$\Pi\hat\tau_\Omegasf^\Y u=\Theta\xi_u$  and  $\Pi\check\tau_\Omegasf^\Y u=\Theta\xi_u$
are local, i.e. they do not couple values of $u$ at 
different points of $\Y$.  }
\end{remark}
\end{subsection}
\begin{subsection}{Approximating non-Friedrichs Dirichlet Laplacians
    by point perturbations.}
Let $\{\Y_N\}_1^\infty$ denote a sequence of discrete sets  $\Y_N=\{\y_k^N\}_1^n\subset\Omegasf$ such that, for any $1\le k\le n$, $$
\y_k^N\equiv(r_k^N\cos\theta_k^N,
r_k^N\sin\theta_k^N)\in\W^{R/3}_k\,,
$$
$$
\inf_{N}\,\sin\beta_k\theta_k^N=c>0
$$
and
$$
\lim_{N\uparrow \infty}\y_k^N=\ve_k\,.
$$  
Posing
$$
\tilde\tau_\Omegasf^{\Y_N}:\D({\Delta_\Omegasf^F})\to\CO^n\,,\quad
\left(\tilde\tau_\Omegasf^{\Y_N} u\right)_k:= 
\frac{u(\y_k^N)}{s_k(\y_k^N)}\,,
$$ 
i.e.
\begin{equation}\label{MN}
\tau_\Omegasf^{\Y_N}=M_N\,\tilde\tau_\Omegasf^{\Y_N}\,,\quad
(M_{N})_{ij}:=s_i(\y^N_i)\,\delta_{ij}\,,
\end{equation}
one has the following 
\begin{lemma}\label{tauN} 
There exist $c>0$ and 
$0<\alpha_k<1-\beta_k$ such that
$$
\left|\left(\tilde\tau^{\Y_N}_\Omegasf u-\tau^\V_\Omegasf u\right)_k\right|\le
c\,\|\y_k^N-\ve_k\|^{\alpha_k}
\| u\|_{\Delta^F_\Omegasf}\,.
$$
\end{lemma}
\begin{proof} 
By 
$$
u(\y^N_k)=u_\circ(\y^N_k)+s(\y^N_k)\ccdot\zeta_u=
u_\circ(\y^N_k)+s_k(\y^N_k)(\zeta_u)_k
\,,
$$ 
by $\tau^\V_\Omega u=\zeta_u$, 
by (\ref{emb2}) and (\ref{cacc}), and denoting by $Q_\circ$ the continuous 
projection 
$Q_\circ:\D(\Delta_\Omegasf^F)\to\D(\Delta_\Omegasf^\circ)$, one 
obtains 
\begin{align*}
&\left|\left(\tilde
\tau^{\Y_N}_\Omegasf u-\tau^{\V}_\Omegasf u\right)_k\right|=
\left|\frac{u_\circ(\y_k^N)}{s_k(\y_k^N)}\right|
\le \frac{1}{c} \,\frac{|u_\circ(\y_k^N)|}{\|\y_k^N-\ve_k\|^
{\beta_k}}\\
\le&\frac{1}{c}\,\|\y_k^N-\ve_k\|^{\alpha-\beta_k}\|u_\circ\|_{H^2(\Omegasf)}
\le{c}\,\|\y_k^N-\ve_k\|^{\alpha-\beta_k}
\|\Delta_\Omegasf u_\circ\|_{L^2(\Omegasf)}\\
\le &{c}\,\|\y_k^N-\ve_k\|^{\alpha-\beta_k}
\|Q_\circ u\|_{\Delta^F_\Omegasf}
\le c\,\|\y_k^N-\ve_k\|^{\alpha-\beta_k}
\| u\|_{\Delta^F_\Omegasf}\,.
\end{align*}
\end{proof}
In conclusion we get the following 
\begin{theorem}\label{convergence}
Let $(\Pi,\Theta)\in\E(\CO^n)$ and define 
$(\Pi_N,\Theta_N)\in\E(\CO^n)$ by
$$
\CO^n_{\Pi_N}=M_N^{-1}(\CO^n_{\Pi})\,,\qquad
\Theta_N=\Pi_NM_N\Theta M_N\Pi_N\,.
$$
Then 
$\Delta_{\Y_N,\Omegasf}^{\Pi_N,\Theta_N}$ converges in norm resolvent
sense to 
$\Delta_{\Omegasf}^{\Pi,\Theta}$ as $N\uparrow\infty$. 
\end{theorem}
\begin{proof}
Let $\tilde \Delta_{\Y^N,\Omegasf}^{\Pi,\Theta}$ be the self-adjoint operator obtained by proceeding as in the proof of Theorem \ref{estensioniY2} with $\tau_\Omegasf^{\Y_N}$ replaced by $\tilde \tau_\Omegasf^{\Y_N}$. Then $\tilde \Delta_{\Y^N,\Omegasf}^{\Pi,\Theta}=A_N^{\Pi,\tilde\Theta_N}$, 
where $\tilde\Theta_N=\Theta-\Pi M_N^{-1}\Lambda^{\Y_N}
M_N^{-1}\Pi$, $\Lambda^{\Y_N}$ is  defined in (\ref{lambdaY}) and the operator sequence $A_N^{\Pi,\tilde\Theta_N}$ is given by Theorem \ref{estensioni} with $A=\Delta_\Omegasf^F$ and $\tau=\tilde\tau^{\Y_N}_\Omegasf$. Analogously (see the proof of Theorem \ref{estensioniOmega2}) $ \Delta_{\Omegasf}^{\Pi,\Theta}=A^{\Pi,\tilde\Theta}$, where $\tilde\Theta=\Theta-\Pi\Lambda\Pi$, $\Lambda$ is  defined in (\ref{lambda}), and $A^{\Pi,\tilde\Theta}$ is given by \ref{estensioni} with $A=\Delta_\Omegasf^F$ and $\tau=\tau^{\V}_\Omegasf$.
\par
Let us now pose $\tilde\chi_i^R:=(2+\beta_i)\pi^{1/2}\chi_i^R/2\beta_iR^{2+\beta_i}$, 
where $\chi_i^R$ denotes the characteristic function of the 
wedge $\W^R_i$. Thus, by the definition of $\tau^\V_\Omega$ and by Lemma \ref{tauN}, 
$$
(\tau^\V_\Omegasf u)_i=\lim_{R\downarrow 0}\,
\langle\tilde\chi_i^{R},u\rangle_{L^2(\Omegasf)}=\lim_{N\uparrow\infty}\,\frac{u(\y_i^N)}{s_i(\y^N_i)}\,.
$$
Then, for any $i\not=j$, by Lemma \ref{G0} and by (\ref{Sweers}) with $w(\y^{N'}_i,\y^N_j)=1$, one has
\begin{align*}
\Lambda_{ij}=&\left(\tau_\Omegasf^\V(g_j)\right)_i=
\lim_{R\downarrow 0}\,\langle\tilde\chi_i^{R},g_j\rangle_{L^2(\Omegasf)}=
\lim_{R\downarrow 0}\,
\left(G_0^*\tilde\chi_i^R\right)_j
\\
=&\lim_{R'\downarrow 0}\,\lim_{R\downarrow 0}\,
\langle\tilde\chi_j^{R'},(-\Delta^F_{\Omegasf})
\tilde\chi_i^R\rangle_{L^2(\Omegasf)}\\
=&
\lim_{R'\downarrow 0}\,\lim_{R\downarrow 0}\,
\int_{\Omegasf\times\Omegasf}\tilde\chi^{R'}_j(\x)g_\Omegasf(0;\x,\y)\tilde\chi^{R}_i(\y)
\,d\x\, d\y\\
=&\lim_{N'\uparrow\infty}\lim_{N\uparrow\infty}\,\frac{g_\Omegasf(0;\y^{N'}_i,\y^N_j)}
{{s_i(\y_i^{N'})}{s_j(\y_j^N)}}\\
=&\lim_{N\uparrow\infty}\,\frac{g_\Omegasf(0;\y^{N}_i,\y^N_j)}
{{s_i(\y_i^{N})}{s_j(\y_j^N)}}\,.
\end{align*}
Thus
$$
\lim_{N\uparrow\infty}\,\frac{\Lambda_{ij}^{\Y_N}}
{{s_i(\y_i^N)}{s_j(\y_j^N)}}=\Lambda_{ij}\,,
$$
i.e. $\tilde\Theta_N$ converges to $\tilde\Theta$. 
Therefore, by Lemma \ref{tauN} and Lemma \ref{taun}, 
$A_N^{\Pi,\tilde\Theta_N}$ converges in norm resolvent sense to 
$A^{\Pi,\tilde\Theta}$, i.e. $\tilde\Delta_{\Y_N,\Omegasf}^{\Pi,\Theta}$ 
converges in norm resolvent
sense to 
$\Delta_{\Omegasf}^{\Pi,\Theta}$.
 The proof is
then concluded by noticing that, by Lemma \ref{tauM}, 
$\tilde\Delta_{\Y_N,\Omegasf}^{\Pi,\Theta}=\Delta_{\Y_N,\Omegasf}
^{\Pi_N,\Theta_N}$. 
\end{proof}
\end{subsection}
\end{section}
\begin{section}{Appendix}
For reader's convenience in this section we collect some
results about the self-adjoint extensions of a closed symmetric operator $S$ with deficiency indices $(n,n)$. Since it suffices for the purposes of this paper we suppose $n<+\infty$ and $-S>0$. For the general case, as well as for the connection with alternative approaches,  we refer to \cite{[Po4]} and references therein.      
\par
Let  
$$
S:\D(S)\subseteq\H\to\H
$$
be a closed symmetric linear operator  on the Hilbert space $\H$
such that $-S>0$. Then by Friedrichs' theorem $S$ 
has a self-adjoint extension $$A:\D(A)\subseteq\H\to\H$$ 
with the same bound. \par 
Suppose now that $S$ has finite deficiency indices $n=n_\pm=$dim$\K_\pm>0$, $\K_\pm:=\K(-S^*\pm i)$. 
By von Neumann's theory of
self-adjoint extensions, there exists an unitary operator $U_A:\K_+\to\K_-$ such that 
$$
\D(A)=\D(S)\oplus_A\G(U_A)\,,
$$
$$
A(\phi+ (\uno+U_A)\phi_+)=S\phi+i\,(\uno-U_A)\phi_+\,,
$$
where $\G(U_A)$ is the graph of $U_A$ and $\oplus_{A}$ denotes the orthogonal sum corresponding to the scalar product inducing the graph norm on $\D(A)$.
Therefore $S=A|\K(P)$, 
where $P:\D(A)\to \K_+$ 
denotes the orthogonal projection onto $\K_+\simeq\CO^n$.
Thus, since this gives some advantages in applications, we will look for the
self-adjoint extensions of $S$ by considering the equivalent 
problem of the search of the self-adjoint extensions of the
restriction of  $A$ to the kernel, 
which we suppose to coincide with $\D(S)$, 
of a surjective bounded linear operator
$$
\tau:\D(A)\to\CO^n \,.
$$
Typically $A$ is an elliptic differential operator and $\tau$ 
is some restriction operator to a discrete set with $n$ points. In the case of infinite defect indices typically $\tau$ is the restriction operator along a null subset and $\CO^n$ is replaced by a (fractional order) Sobolev-Hilbert space (see \cite{[Po1]}, \cite{[Po4]} and references therein). \par
By \cite{[Po1]}, \cite{[Po2]} and \cite{[Po4]} one has the following
\begin{theorem}\label{estensioni} 
\par\noindent The set of all self-adjoint extensions of $S$ is parametrized by the bundle $p:\E(\CO^n)\to\P(\CO^n)$; if $A^{\Pi,\Theta}$ denotes the self-adjoint extension corresponding to $(\Pi,\Theta)\in \E(\CO^n)$ 
 then 
\begin{equation*}
A^{\Pi,\Theta}:\D(A^{\Pi,\Theta})\subseteq\H\to\H\,,\quad
A^{\Pi,\Theta}\phi:=A\phi_0\,,
\end{equation*}
\begin{align*}
\D(A^{\Pi,\Theta})
:=\left\{\phi=\phi_0+G_0
\xi_\phi\,,\ \phi_0\in \D(A)\,,\, \xi_\phi\in\CO^n_\Pi
\,,\, \Pi\tau
\phi_0=\Theta\xi_\phi
\right\}\,,
\end{align*}
where 
$$
G_z:\CO^n\to\H\,,\quad G_z:=\left(\tau(-A+\bar z)^{-1}\right)^*\,,\quad z\in\rho(A)\,.
$$
Moreover the resolvent of $A^{\Pi,\Theta}$ is given, for any 
$z\in\rho(A)\cap\rho(A^{\Pi,\Theta})$, by Kre\u\i n's type formula
$$
(-A^{\Pi,\Theta}+z)^{-1}=(-A+z)^{-1}+G_z
\Pi(\Theta+z\Pi G_0^*G_z\Pi)^{-1}\Pi G^*_{\bar z}\,.
$$
\end{theorem}

\begin{remark}{\rm
One can easily check that the density hypothesis about $\K(\tau)$
gives, for any $z\in \rho(A)$,
\begin{equation*}
\R(G_z)\cap \D(A)=\{0\}\,,
\end{equation*}
thus the decomposition appearing in $\D(A^{\Pi,\Theta})$ is well defined. Moreover, since by first resolvent identity
\begin{equation}\label{GR1}
(z-w)(-A+w)^{-1}G_z=G_w-G_z\,,\\
\end{equation}
one has
\begin{equation*}
\R(G_w-G_z)\subset \D(A)
\end{equation*}
and  
\begin{equation}\label{gamma1}
z\,G_0^* G_{z}=\tau(G_0-G_z)
\end{equation}
}
\end{remark}
\begin{remark} {\rm Notice that the knowledge of 
the adjoint $S^*$ is not required. However it can be
readily calculated: by \cite{[Po3]}, Theorem 3.1, one has 
$$
S^*:\D(S^*)\subseteq\H\to\H\,,\qquad S^*\phi=A\phi_0\,,
$$
$$
\D(S^*)=\{\phi\in\H\,:\,\phi=\phi_0+G_0\xi_{\phi},\  \phi_0\in
\D(A),\ \xi_{\phi}\in\CO^n\}\,.
$$
Moreover
$$
\D(A^{\Pi,\Theta})=
\{\phi\in\D(S^*)\,:\,\rho_0\phi\in\CO^n_\Pi\,,\ 
\Pi\tau_0\phi=\Theta \rho_0\phi\}\,.
$$
where the regularized trace operators $\tau_0$ and $\rho_0$ 
are defined by
$$
\tau_0:\D(S^*)\to\CO^n\,,\qquad\tau_0\phi :=\tau 
\phi_0
$$
and 
$$
\rho_0:\D(S^*)\to\CO^n\,,\qquad\rho_0\phi :=\xi_\phi\,.
$$
By \cite{[Po3]}, Theorem 3.1, $(\CO^n,\tau_0,\rho_0)$ is a boundary
triple for $S^*$, with corresponding Weyl 
function $zG^*_0G_Z$, and the Green-type formula
\begin{equation}\label{Green}
\langle\phi,S^*\psi\rangle-\langle S^*\phi,\psi\rangle
=\tau_0\phi\ccdot\rho_0\psi
-\rho_0\phi\ccdot\tau_0\psi
\end{equation}
holds true. Also notice that $G_z\xi$ solves the boundary value type 
problem
\begin{align}\label{Gz}
\begin{cases}
S^*G_z\xi=zG_z\xi\,,&\\
\rho_0G_z\xi=\xi\,.&
\end{cases}
\end{align}
Here we refer to \cite{[DM1]} for boundary triplets theory.}
\end{remark}
Since $S=A|\K(\tau)=A|\K(\tau_M)$, 
where $\tau_M:=M\tau$ and $M:\CO^n\to\CO^n$ 
is any bijective linear map, the group $\text{GL}(\CO^n)=\{M\,:\,\det(M)\not=0\}$ acts on 
the  bundle $\E(\CO^n)$ by 
$$
\alpha:\text{GL}(\CO^n)\times\E(\CO^n)\to\E(\CO^n)\,,\quad 
\alpha(M,(\Pi,\Theta))=(\Pi_M,\Theta_M)\,,
$$
where $\alpha$ is defined in such a way that 
$$
A_M^{\Pi_M,\Theta_M}=A^{\Pi,\Theta}\,,
$$
with $A_M^{\Pi,\Theta}$ denoting the extension corresponding to 
$(\Pi,\Theta)\in \E(\CO^n)$ provided by Theorem 
\ref{estensioni} in the case one uses the map $\tau_M$. The action $\alpha$ is explicitly given in the following 
\begin{lemma}\label{tauM}  
$$\CO^n_{\Pi_M}=(M^*)^{-1}(\CO^n_{\Pi})\,,\qquad
\Theta_M=\Pi_MM\Theta M^*\Pi_M\,.
$$
\end{lemma}
\begin{proof} By the definitions of $\tau_M$ and $G_0$ one has that any $\phi\in \D(A_M^{\Pi_M,\Theta_M})$ is of the kind 
$\phi=\phi_0+G_0M^*\zeta$, where $\Pi_M M\tau\phi_0
=\Theta_M\zeta$, $\zeta\in\CO^n_{\Pi_M}$, i.e. 
$\phi=\phi_0+G_0\xi$, where $\Pi_MM\tau\phi_0
=\Theta_M(M^*)^{-1}\xi$, $\xi\in\CO^n_{\Pi}$. 
Since $(\CO^n_{\Pi_M})^\perp=M((\CO^n_\Pi)^\perp)$, 
$\Pi_M M\Pi:\CO^n_\Pi\to\CO^n_{\Pi_M}$ is a bijection. Thus 
$\Pi_MM\tau\phi_0=\Theta_M(M^*)^{-1}\xi$ if and only if
$\Pi\tau\phi_0=(\Pi_M M\Pi)^{-1} \Theta_M(M^*)^{-1}\xi$.
\end{proof} 
We conclude by providing a simple convergence criterion:
\begin{lemma}\label{taun} Given the bounded operators $\tau_N:\D(A)\to\CO^n$ and  $\tau:\D(A)\to\CO^n$, consider the symmetric operators $S_N=A|\K(\tau_N)$ and $S=A|\K(\tau)$. Given $\Pi\in \P(\CO^n)$ and the symmetric operators $\Theta_N:\CO^n_\Pi\to\CO^n_\Pi$ and $\Theta:\CO^n_\Pi\to\CO^n_\Pi$, let $A_N^{\Pi,\Theta_N}$ and $A^{\Pi,\Theta}$ be the self-adjoint extensions of $S_N$ and $S$ given by Theorem \ref{estensioni}. If $\tau_{N}$ norm-converges to $\tau$ and $\Theta_{N}$ converges to $\Theta$ as $N\uparrow\infty$, then 
$A_N^{\Pi,\Theta_N}$ converges in
norm resolvent sense to $A^{\Pi,\Theta}$.
\end{lemma}
\begin{proof}
By our hypothesis on $\tau_N$,
$G_{N,z}^*:= \tau_N (-A+z)^{-1}$ and $G_{N,z}$ 
norm-converge to
$G_z^*$ and $G_z$ respectively. 
This implies that $zG_{N,0}^*G_{N,z}$ norm-converge to 
$zG_{0}^*G_{z}$ and hence $(\Theta_N+z\Pi G_{N,0}^*G_{N,z}\Pi)^{-1}$
norm-converge to 
$(\Theta+z\Pi G_{0}^*G_{z}\Pi)^{-1}$. The thesis then follows by
the resolvent formula provided in Theorem \ref{estensioni}. 
\end{proof}
\end{section}

\end{document}